\documentclass[12pt]{extarticle}
\usepackage{extsizes}
\usepackage[reqno]{amsmath} 
\usepackage{amsmath}
\usepackage{amssymb}
\usepackage{amsthm}
\usepackage{amscd}
\usepackage{amsfonts}
\usepackage{graphicx}
\usepackage{dsfont}
\usepackage{fancyhdr}
\usepackage{enumerate}
\usepackage{youngtab}
\usepackage[utf8]{inputenc}
\usepackage{comment}
\usepackage{cite}

\DeclareMathOperator{\den}{\mathrm{den}}
\DeclareMathOperator{\lcm}{\mathrm{lcm}}

\usepackage{subfigure}
\usepackage[margin=1.5in]{geometry}
\usepackage{graphicx}
\usepackage{algorithm}
\usepackage{algpseudocode}
\usepackage{color}
\usepackage{hyperref}
\usepackage{notation}

\theoremstyle{theorem}
\newtheorem{corollary}{Corollary}

\newtheorem{lemma}[corollary]{Lemma}

\newtheorem{theorem}[corollary]{Theorem}

\begin{document}

\title{Quantum Walks in the Normalized Laplacian}
\author{Gabriel Coutinho$^*$ \and Pedro Ferreira Baptista\footnote{Both authors at the Department of Computer Science, UFMG - Belo Horizonte, Brazil. [gabriel, pedro.baptista@dcc.ufmg.br] }}
\date{\today}
\maketitle
\vspace{-0.8cm}

\begin{abstract} 

We provide a characterization of perfect state transfer in a quantum walk whose Hamiltonian is given by the normalized Laplacian. We discuss a connection between classical random walks and quantum walks only present in this model, and we also rule out several trees as candidates to host perfect state transfer.
\end{abstract}

\section{Introduction}

In this paper, we look at finite, undirected graphs without weights on $n$ vertices. We also assume that the graph is connected. We perform a quantum walk on the 1-excitation subspace of $G$ with respect to a Hamiltonian given by the normalized Laplacian matrix $\mathcal{L}(G)$. For a more explicit description of the global Hamiltonian resulting in this quantum walk, we refer to \cite[
Theorem 2.1]{chaves2022and}.

The normalized Laplacian has been extensively studied in \cite{FanChungSGT}, where, among other things, it is shown to be a natural approach to describe several physical processes in a graph, including random walks and heat kernels. On the other hand, the spike in interest in studying quantum walks in graphs over the past twenty years or so (see for instance \cite{childs2009universal,kendon2011perfect,kay2010perfect,chakraborty2016spatial,wong2016laplacian,venegas2012quantum,CoutinhoGodsilSurvey,kadian2021quantum} and references therein) has overlooked the usage of the normalized Laplacian as a suitable Hamiltonian, with the sole exception of \cite{alvir2016perfect}. This latter work mainly discusses quantum walks modeled by Hamiltonians which correspond to the (combinatorial) Laplacian, but it contains a result that a Hamiltonian based on the normalized Laplacian would not yield perfect state transfer in paths (also called lines).

One reason for an apparent difficulty in dealing with the normalized Laplacian is that its entries are typically not integers, and therefore its spectrum is not made out of algebraic integers. A main method \cite{GodsilPerfectStateTransfer12} to study state transfer in graphs seems to require these facts. Our first and main result in this paper is a proof that a similar result holds for the normalized Laplacian. This is shown in Section \ref{secIntronormalizedLaplacian}.

As an application, we achieve partial success in ruling out perfect state transfer in trees according to the normalized Laplacian, in Section \ref{secpstTreenormalizedLaplacian}. State transfer in trees is a common theme of interest, as trees are good candidate architectures for quantum networks. In \cite{Coutinho2015}, it is shown that state transfer does not happen in trees for the quantum walk model described by the combinatorial Laplacian, and it is conjectured that the analogous result holds for the adjacency matrix model. Our work in this paper suggests that there could be a deeper reason, related to graph connectivity and which transcends the Hamiltonian model, that prevents state transfer from happening in poorly connected graphs.

We start our paper by drawing attention to the concept of cospectrality within the normalized Laplacian framework. This concept plays a key role in understanding the combinatorics of quantum walks in the adjacency matrix model, but in contrast, it is poorly understood in the combinatorial Laplacian model. For the normalized Laplacian, this concept will bring out an interesting and new connection between quantum walks and classical random walks.

We finish our paper with some discussion of other results and a conjecture.

\section{Combinatorics of normalized Laplacian cospectrality}

Let $G$ be a connected graph on two or more vertices. Let $A$ be its adjacency matrix and $D$ the diagonal matrix with the degrees of its vertices in its entries. The combinatorial Laplacian matrix is defined as $L = D - A$, and the normalized Laplacian as 

\[\mathcal{L} = D^{-1/2} L D^{-1/2}.\]

As $\mathcal{L}$ is symmetric, it admits a spectral decomposition, denoted by $\mathcal{L} = \sum_{r = 0}^d \theta_r F_r$, and therefore the quantum walk is given by

\begin{align}
	U(t) = \exp(\ii t \mathcal{L}) =  \sum_{r = 0}^d \e^{\ii t\theta_r} F_r, \label{eq:sepcdecomp}
\end{align}\noindent
where $\theta_r$ are the eigenvalues and $F_r$ the eigenprojection matrices of $\mathcal{L}$.

Perfect state transfer between vertices $a$ and $b$ implies that $(F_r)_{a,a} = (F_r)_{b,b}$ for all $r$ (we will see more details about this later). When this occurs, we say that the vertices $a$ and $b$ are cospectral relative to the normalized Laplacian.

The idempotents $\{F_r\}$ form a basis for the adjacency algebra generated by $\mathcal{L}$, therefore
\[
	(F_r)_{a,a} = (F_r)_{b,b} \text{ for all }r \iff (\mathcal{L}^k)_{a,a} = (\mathcal{L}^k)_{b,b}\text{ for all }k \geq 0.
\] 

This has been shown in details for the adjacency matrix in \cite[Theorem 3.1]{godsil2017strongly}, but the same proof works here. 

For the adjacency matrix $A$, the definition of cospectrality carries a combinatorial interpretation; for $a, b \in V(G)$, $A^k_{a,b}$ is equal to the number of different walks of length $k$ between $a$ and $b$. So, two vertices $a$ and $b$ being cospectral means that for any integer $k$ the number of closed walks, i.e., walks that start and end at the same vertex, of size $k$ is the same for $a$ and $b$. A combinatorial interpretation for normalized Laplacian cospectrality is also available.

First, we recall the relation between the normalized Laplacian and classical random walks in graphs. In a random walk, each vertex admits a probability distribution on its edges. Then, starting at a vertex, at each discrete iteration, we jump from the current vertex to another vertex taken from a sampling of the edges of the current vertex. In the simplest case, the probability distribution is uniform, therefore the transition matrix of the random walk is given by $D^{-1}A$. More generally, if $v$ is a probability distribution over the vertices, $v^T (D^{-1}A)^k $ gives the probability distribution after $k$ steps of the random walk.

From the definition of the normalized Laplacian and by conjugating $\mathcal{L}$ by $D^{-1/2}$, we get that 

\[
    D^{-1/2} \mathcal{L} D^{1/2} = D^{-1}L,
\]

\noindent
which can be rearranged into

\[
    D^{-1}A = D^{-1/2} (I - \mathcal{L}) D^{1/2}.
\]

If we start at a vertex $a$, the probability that we end at that vertex after $k$ steps of the random walk is $[(D^{-1}A)^k]_{a,a}$. The connection above will show that normalized Laplacian cospectrality between $a$ and $b$ is equivalent to having that this probability is the same for both vertices.

In this paper, we opt to use the bra-ket notation, thus the characteristic vector of vertex $a$ is denoted by $\ket a$ and the corresponding dual functional by $\bra a$.

\begin{lemma}
\label{lemmaCospectralitynormalizedLaplacian}
Let $G$ be a connected graph with at least two vertices. Vertices $a$ and $b$ are cospectral with respect to the normalized Laplacian if and only if, for any $k \geq 0$, $[(D^{-1} A)^k]_{a,a} = [(D^{-1} A)^k]_{b,b}$.
\begin{proof}
We know that $D^{-1}A = D^{-1/2} (I - \mathcal{L}) D^{1/2}$. Thus, if $\ket a$ stands for the characteristic vector of the vertex $a$, we have
\begin{align*}
    [(D^{-1}A)^k]_{a,a} &= [(D^{-1/2} (I - \mathcal{L}) D^{1/2})^k]_{a,a}\\
    &= \bra a D^{-1/2} (I - \mathcal{L})^k D^{1/2} \ket a \\
    &=  [(I - \mathcal{L})]^k_{a,a}.
\end{align*}

We already saw that $a$ and $b$ are cospectral if and only if $(\mathcal{L}^k)_{a,a} = (\mathcal{L}^k)_{b,b}$ for all $k$. The conclusion now follows from the fact that $(I - \mathcal{L})^k$ is a polynomial in $\mathcal{L}$, and thus $[(I - \mathcal{L})]^k_{a,a} = [(I - \mathcal{L})]^k_{b,b}$.
\end{proof}
\end{lemma}

\newcommand{\LL}{\mathcal{L}}
\section{State transfer and periodicity characterization}
\label{secIntronormalizedLaplacian}

We assume the graph $G$ is connected on at least two vertices.

We start by recalling some basic properties of the matrix $\LL$, which we will make use without reference. See \cite{FanChungSGT} for a complete treatment and proofs of these and other basic facts. The normalized Laplacian is positive semidefinite, and the multiplicity of $0$ as an eigenvalue is equal to the number of connected components of the graph. If the graph is connected, the unique (up to scaling) eigenvector in the $0$ eigenspace is non-zero everywhere with entries of the same sign (in fact a multiple of $D^{1/2} \ket{\mathds{1}}$, $\ket{\mathds{1}} $ the all-ones vector).

We typically assume that the eigenvalues are ordered as $0 = \theta_0 < \theta_1 \leq  ... \leq \theta_{n-1}$. It is immediate to verify that the eigenvalues sum to $n$, and this implies that $\theta_1 \leq \frac{n}{n-1}$, with equality if and only if $G$ is the complete graph. In fact, if $G$ is not complete, then $\theta_1 \leq 1$, and in general, for all $i \leq n - 1$, we have       
\[
	\theta_i \leq 2,
\]
with $\theta_{n-1} = 2$ if and only if $G$ is bipartite. In this case, the entire spectrum is symmetric around $1$.

Recall that the quantum walk transition matrix is given by $U(t) = \exp(\ii t \mathcal{L})$. We say that perfect state transfer happens between vertices $a$ and $b$ at time $\tau$ if there is a complex unit $\lambda$ so that
\[
	U(\tau) \ket a = \lambda \ket b.
\]
The first thing to notice is that, upon decomposing the transition matrix into spectral idempotents (as in Equation~\eqref{eq:sepcdecomp}), we have that perfect state transfer is equivalent to having, for all eigenvalues $\theta_r$,
\begin{align}
	\e^{\ii t\theta_r} F_r \ket a = \lambda F_r \ket b. \label{eq:1}
\end{align}
As $F_0$ is real and all positive, it must be that $1 = \e^{\ii t\theta_0} = \lambda$. Also, $U(t)$ is symmetric, and therefore
\[
	U(\tau) \ket a = \ket b \implies U(2\tau) \ket a = U(\tau) \ket b = \ket a.
\]
That is, vertex $a$ is periodic at time $2\tau$.

Further, Equation~\eqref{eq:1} is equivalent to having $F_r \ket a = \pm F_r \ket b$, and $\e^{\ii t\theta_r} = \pm 1$, with signs matching for each index $r$. The former condition is called strong cospectrality between vertices $a$ and $b$, and readily implies as a consequence that $a$ and $b$ are cospectral, that is, $(F_r)_{a,a} = (F_r)_{b,b}$. See \cite{Godsil2008,Godsil2010,CoutinhoPhD} for an introductory treatment of these facts to the quantum walk given by the adjacency matrix $A(G)$ but which also applies for any Hermitian matrix instead.

We say that an eigenvalue $\theta_r$ is in the eigenvalue support of a vertex $a$ if $E_r \ket a \neq 0$. Whether or not there exists a time for which a vertex is periodic can be completely determined by an integrality condition on its eigenvalue support. 

For any real numbers $\theta_i, \theta_j, \theta_k, \theta_l$, with $\theta_k \neq \theta_l$, we say that the ratio condition holds for them if

\[
    \frac{\theta_i - \theta_j}{\theta_k - \theta_l} \in \mathbb{Q}.
\]

\begin{theorem}[{\cite[Lemma 2.2]{Godsil2008}}]
\label{Theoremratiocondition}
A graph $G$ is periodic at a vertex $a$ with respect to $\LL$ if and only if the ratio condition holds for the eigenvalues of $\LL$ in the support of vertex $a$. 
\end{theorem}

A stronger characterization of the ratio condition is available.

Let $\alpha$ be a real number, then if it is a root of a rational polynomial, we say that it is an algebraic number. If the polynomial has integer coefficients and is monic we say that it is an algebraic integer. Moreover, if the polynomial has degree two, we say that it is a quadratic irrational and quadratic integer, respectively.

The following comes from \cite[Theorem 6.1]{Godsil2010} and \cite[Lemma 5.4]{chan2019quantum}, upon realizing that any set of algebraic numbers can be scaled with a common factor to a set of algebraic integers.

\begin{lemma}
\label{lemma:ratioalgebraicnumbers}
Let $S = \{\theta_0, \theta_1, ..., \theta_d\}$ be a set of real algebraic numbers, closed under taking algebraic conjugates. Then, the ration condition holds for $S$ if and only if either condition holds:
\begin{enumerate}
    \item The elements in $S$ are rational.
    \item The elements in $S$ are quadratic irrational. Moreover, there is a square-free integer $\Delta > 1$ and rational numbers $a, b_0, b_1, ..., b_d$ so that
    \[
        \theta_r = \frac{1}{2} (a + b_r \sqrt{\Delta}).
    \]
\end{enumerate} \qed
\end{lemma}

Our first result is that, for the normalized Laplacian, the ratio condition holding for the eigenvalues in the support of a vertex implies that these eigenvalues are rational.

Let $\den(\cdot)$ be a function that returns the denominators of a rational number. If the number is an integer, then it returns $1$.

\begin{theorem}
\label{theoremperiodicityeigenvalues}
Suppose $G$ is a connected graph with at least two vertices, $\mathcal{L}$ its normalized Laplacian and let $S = \{\theta_0, \theta_1, ..., \theta_d\}$ be the distinct eigenvalue support of the vertex $a$, with $\theta_0 < \theta_1 < ... < \theta_d$. Then $G$ is periodic at $a$ if and only if the eigenvalues in $S$ are rational. Moreover, if the condition holds, let

\[
m = \lcm \{ \den (\theta_r ) \}_{\theta_r \in S}
\]
and 
\[
g = \gcd \{ m \theta_r\}_{\theta_r \in S}.
\]
Then the smallest positive $\tau$, such that $U(\tau) \ket a = \ket a$ is 
\[
\tau = \frac{2 m \pi}{g}    
\]
and any other time periodicity at $a$ occurs is an integer multiple of $\tau$.
\end{theorem}

\begin{proof}
The first thing to notice is that as $G$ is connected and has at least two vertices, then there are at least two eigenvalues in the support of $a$, for otherwise $\LL \ket a$ would be a multiple of $\ket a$.

If the eigenvalues are rational, then the ratio condition holds and by Theorem \ref{Theoremratiocondition} the vertex is periodic.

First note that the eigenvalues of $\mathcal{L}$ are algebraic numbers --- for instance, $\mathcal{L}$ is similar to $D^{-1} L$ which is a rational matrix, thus its eigenvalues are all roots of a polynomial with rational coefficients.

In order to apply Lemma~\ref{lemma:ratioalgebraicnumbers} we need to show that $S$, the eigenvalue support of $a$, is closed under taking algebraic conjugates. Let $\ket \theta_r $ be an eigenvector of $\mathcal{L}$ for the eigenvalue $\theta_r$. Let also $\theta_s$ be an algebraic conjugate of $\theta_r$. It is straightforward to show that $D^{-1/2} \ket \theta_r$ is an eigenvector of $D^{-1}L$ for the eigenvalue $\theta_r$, as

\begin{equation}
    D^{-1}LD^{-1/2}\ket{\theta_r} = D^{-1/2}\mathcal{L} D^{1/2} D^{-1/2} \ket{\theta_r} = D^{-1/2} \mathcal{L} \ket{\theta_r} = \theta_r D^{-1/2}\ket{\theta_r}.
\end{equation}

Now, we show that a field automorphism $\psi$ that sends $\theta_r$ to $\theta_s$ also sends $D^{-1/2}\ket{\theta_r}$ to some eigenvector of $\theta_s$.

\begin{align*}
    \psi(\theta_r) \psi(D^{-1/2} \ket{\theta_r}) & = \psi(D^{-1} L D^{-1/2} \ket{\theta_r}) \\ & = \psi(D^{-1} L) \psi(D^{-1/2} \ket{\theta_r}) \\ & = D^{-1} L \psi(D^{-1/2} \ket{\theta_r}).
\end{align*}

So, $E_r \ket a = 0$ if and only if $E_s \ket a = 0$, then the eigenvalue support is closed under taking conjugates.

We can then apply Lemma~\ref{lemma:ratioalgebraicnumbers}. Because $\theta_0 = 0$ is in the support of all vertices, it follows that if the eigenvalues in the support of $a$ are quadratic irrational in the form given by Lemma~\ref{lemma:ratioalgebraicnumbers}, then $a = 0$. But the conjugates of multiples of square roots are their negatives, and as the normalized Laplacian is positive semidefinite all its eigenvalues are non-negative. Thus, the eigenvalues in the eigenvalue support of $a$ are all rational.

For the final part, note that if a vertex $a$ is periodic at a time $\tau$, we can write $\tau$ as 

\begin{equation}
    \tau = t \frac{2\pi m}{g},
\end{equation}
where $t$ is some real number. One consequence of Theorem \ref{Theoremratiocondition} is that periodicity at time $\tau$ implies that there are integers $m_{r,s}$ such that $\tau(\theta_r - \theta_s) = 2 m_{r,s} \pi$ for $\theta_r, \theta_s \in S$. By fixing $\theta_s = \theta_0$, we can rewrite this equation as

\begin{equation}
    \tau = \frac{2 m_{r,0} \pi}{\theta_r},
\end{equation}

\noindent
for $\theta_r \neq 0$. Now, using both equations we defined for $\tau$, we can see that

\begin{equation}
    t \frac{m \theta_r}{g} \in \mathbb{Z}.
\end{equation}

Thus, from the definitions we gave for $m$ and $g$, we can see that $t$ must be an integer.
\end{proof}

With these results, we can give a characterization of perfect state transfer for the normalized Laplacian Hamiltonian model. The proof follows the lines of the characterization for the adjacency matrix model \cite[Theorem 2.4.4]{CoutinhoPhD} given the results obtained above, therefore we omit it.

\begin{theorem}
\label{theorempstnormalizedlaplacian}
Let $a$ and $b$ be vertices in a connected graph $G$, $U(t) = exp(\ii t \mathcal{L})$ express the unitary time evolution of our system, $\mathcal{L}$ being the normalized Laplacian. Let $S = \{\theta_0, \theta_1, ..., \theta_d\}$ be the eigenvalue support of $a$, with $0 = \theta_0 < \theta_1 < ... < \theta_d$. Then, there is perfect state transfer between $a$ and $b$ if and only if the following conditions hold:

\begin{enumerate}
    \item Vertices $a$ and $b$ are strongly cospectral.
    \item The eigenvalues in $S$ are rational.
    \item Let

\begin{equation}
    m = \lcm \{ den (\theta_r)\}_{\theta_r \in S} 
\end{equation}

\noindent
and 

\begin{equation}
    g = \gcd \{ m \theta_r\}_{\theta_r \in S}.
\end{equation}

For all $0 \leq r \leq d$, the following holds:

\begin{itemize}
    \item $(E_r)_{a,b} > 0$ if and only if $m\theta_r/g$ is even,
    \item $(E_r)_{a,b} < 0$ if and only if $m\theta_r/g$ is odd.
\end{itemize}
 
If the conditions hold, then the minimum positive time we have perfect state transfer between $a$ and $b$ is $\tau = m \pi/g$, and any other time it occurs is an odd multiple of $\tau$.
\end{enumerate} \qed
\end{theorem}

\section{Perfect state transfer in trees}
\label{secpstTreenormalizedLaplacian}

Trees are minimally connected graphs, thus good choices if one wants to build a network and minimize resources. So, it is not surprising that research has been done in state transfer in trees.

In \cite{Coutinho2015}, it was shown that no tree on three or more vertices admits perfect state transfer with regard to the Laplacian matrix Hamiltonian model. It is still an open problem whether there are trees on more than 3 vertices that present perfect state transfer with respect to the adjacency matrix model.

So, considering how the normalized Laplacian can be defined in terms of the Laplacian matrix, we ask if a similar result would hold for the normalized Laplacian.

Recall from our discussion Section~\ref{secIntronormalizedLaplacian} that if $G$ is a tree, then $0$ and $2$ are eigenvalues of $\LL(G)$ with multiplicity $1$, and the spectrum is symmetric about $1$. Moreover, each eigenvector corresponding to $\theta$ give an eigenvector corresponding to $2-\theta$ obtained upon switching the sign of entries of the vertices in one of the bipartite classes.

To facilitate our notations, assuming $a$ and $b$ are strongly cospectral vertices with eigenvalue support $S$, let us subdivide $S$ into positive $S^+$ and negative $S^-$ eigenvalue supports. The former containing the eigenvalues $\theta_r$ that have the same projection onto the eigenspace, i.e., $E_r \ket a = E_r \ket b$, and the latter defined analogously. If perfect state transfer happens between $a$ and $b$, by Theorem \ref{theorempstnormalizedlaplacian}, $\theta_r \in S^+$ if and only if $m \theta_r/g$ is even for $m$ and $g$ as defined in the theorem. Similarly, $\theta_r \in S^-$ if and only if $m \theta_r/g$ is odd.

Also, we can show that for two strongly cospectral vertices, neither $S^-$ and $S^+$ are empty. Moreover, if the graph has at least three vertices, then the following result, originally proved for the combinatorial Laplacian, also holds.

\begin{lemma}\cite[Lemma 3.1]{Coutinho2015}
\label{lemmasupportisnonempty}
Assume $G$ has at least three vertices, and let $a$ and $b$ be strongly cospectral vertices with respect to $\LL$. Then, if $S^+, S^-$ are their positive and negative eigenvalue support, respectively, then $|S^+| \geq 2$ and $|S^-| \geq 1$.
\end{lemma}

Using the known facts about eigenvalues of $\LL$ for bipartite graphs, we are able to provide quite restrictive conditions to the eigenvalue support of vertices involved in state transfer in trees according to $\LL$.

\begin{theorem}
\label{theorem:rationaleigenvaluestree}
Let $G$ be a tree on two or more vertices admitting perfect state transfer between vertices $a$ and $b$. Let $S^-, S^+$ be their negative and positive eigenvalue support, respectively. If $a$ and $b$ are on the same bipartite class, then $|S^- | = 1$. If $a$ and $b$ are on different bipartite classes and if $\theta \in S^-$, then $\theta = 2/k$ for some odd positive integer $k$.
\end{theorem}
\begin{proof}
 From Theorem~\ref{theorempstnormalizedlaplacian}, we know that all eigenvalues in the eigenvalue support of $a$ and $b$ are rational. 
 
 We assume $\theta = p/q$ is a rational eigenvalue of $\LL$ in $S^-$, with $p \geq 0$ and $q \geq 1$ coprime.

The matrix $\mathcal{L}$ is similar to the matrix $D^{-1} L$, and $\mathcal{L} \ket z = \theta \ket z$ if and only if $D^{-1}LD^{-1/2} \ket z = \theta D^{-1/2} \ket z$. As $D^{-1} L $ is a rational matrix, we can assume $\ket z$ is an eigenvector for $\theta$ so that $\ket w = D^{-1/2} \ket z$ is an integer vector with coprime entries. Note that ${L \ket w = \theta D \ket w}$.

Let $d$ be a leaf and $c$ its unique neighbor. Then we can see that

\[
   \braket{d}{w} - \braket{c}{w} = \bra d L \ket w = \frac{p}{q} \braket{d}{w}.
\]

The left-hand side is the sum of two integers, therefore, as $\gcd\{p,q\} = 1$, $q \mid \braket{d}{w}$. This means that

\[    \braket{d}{w} - \braket{c}{w} \equiv \theta \braket{d}{w} \equiv 0 \mod{p}.
\]

Therefore, $\braket{d}{w} \equiv \braket{c}{w} \mod{p}.$ Since $G$ is a tree we can proceed recursively from the leafs to their unique neighbors, at each step considering vertices that have exactly one neighbor not considered before. We conclude that all entries of $\ket w$ are equivalent $\mod p$. 

With $\theta \in S^-$, it must be that

\[    \braket{a}{z} = -\braket{b}{z}.
\]

Since $D^{-1/2}$ is a non-negative matrix and $\ket w = D^{-1/2} \ket z$, we have

\[    \braket{a}{w} \equiv - \braket{b}{w} \mod p.
\]

As we assumed that the greatest common divisor of the entries of $\ket w$ is $1$, this can only happen if $p = 1$ or $p = 2$. 

Now, assume $a$ and $b$ are from the same bipartite class. As $\theta = p/q \in S^-$ with $p = 1$ or $2$, we have that $2 - \theta$ is an eigenvalue whose eigenvector entries will remain unchanged in both $a$ and $b$, that is

\[    2 - \theta = \frac{2q - p}{q} \in S^{-},
\]
and as $2q-p$ and $q$ are coprime, this means that $2q - p = 1$ or $2q - p = 2$. If $p = 1$, then $q = 1$ or $3/2$. Since $q \in \mathbb{Z}$, then $p/q = 1$. Now, if $p = 2$, then $q = 3/2$ or $q = 2$, also giving $p/q = 1$. Therefore, if $a$ and $b$ are in the same bipartite class, $p/q$ can only be equal to $1$ and from Lemma \ref{lemmasupportisnonempty} we get that $|S^-| = 1$.

Suppose now $a$ and $b$ are in different bipartite classes, thus $2 - \theta \in S^+$. As we assumed that perfect state transfer happens between $a$ and $b$, then for $m$ and $g$ as defined in Theorem \ref{theorempstnormalizedlaplacian}, we have

\begin{align*}
        \frac{m}{g}\cdot\frac{p}{q} \quad \text{is odd} \text{ and} \quad  \frac{m}{g}\cdot\frac{(2q - p)}{q} \quad \text{is even}. 
\end{align*}

If $p = 1$, $m/(gq)$ odd implies that $2q - p$ is even. But, as $p$ is odd, this is not possible. 

The only possible solution is that $p = 2$; in that case, since we assumed that $\gcd (p,q) = 1$, we have that $q$ is odd.
\end{proof}

Let $a$ and $b$ be vertices and have $N(a)$ and $N(b)$ denote the set of neighbors of $a$ and $b$, respectively. Then $a$ and $b$ are said to be twins if either $N(a) = N(b)$ or $a \cup N(a) = b \cup N(b)$.

It was shown in \cite{Coutinho2015} that if $a$ and $b$ are strongly cospectral vertices for the combinatorial Laplacian and $|S^-| = 1$, then $a$ and $b$ are twins. The analogous result holds for the normalized Laplacian, and therefore we get the following corollary.

\begin{corollary}
\label{cor:negativesupporttwins}
Let $G$ be a tree on three or more vertices, admitting perfect state transfer between vertices $a$ and $b$. If $a$ and $b$ are in the same bipartite class, then they are leafs attached to the same vertex.
\end{corollary}

\begin{proof}
From Theorem~\ref{theorem:rationaleigenvaluestree}, we know that $S^- = \{1\}$. If $\mathcal{L} = \sum_r \theta_r F_r$ is the spectral decomposition of the normalized Laplacian, we define

\begin{equation*}
    z^+ = \sum_{\theta_r \in S^+} F_r e_a \text{,} \hspace{1cm} z^- = \sum_{\theta_r \in S^-} F_r e_a.
\end{equation*}

As $a$ and $b$ are strongly cospectral, it is easy to see that $z^- = \frac{1}{2}(\ket a - \ket b)$, and because $|S^-| = 1$, it follows that $\ket a - \ket b$ is an eigenvector for the eigenvalue $1$. Moreover, $D^{1/2}\ket{\mathds{1}}$ is an eigenvector for the eigenvalue $0 \in S^+$, so the entries corresponding to $a$ and $b$ are equal, therefore $\deg(a) = \deg(b)$. It is then straightforward to check that $a$ and $b$ have the same set of neighbors, and because the graph is acyclic on at least three vertices, they are leafs attached to the same common neighbor.
\end{proof}

It is easy to check that the path on two vertices, $P_2$, admits perfect state transfer between its only two vertices also in the normalized Laplacian model. This is the only case in which vertices in different bipartite classes are involved in state transfer if $|S^-| = 1$.

\begin{corollary}
Let $G$ be a tree with two or more vertices. Assume that $a$ and $b$ are vertices in different bipartite classes, and such that perfect state transfer happens between them. If $2$ is the only eigenvalue in $S^-$, then $G = P_2$.
\end{corollary}
\begin{proof}
From the proof of Corollary~\ref{cor:negativesupporttwins}, we know that if $|S^-| = 1$, then $a$ and $b$ are twins. Since $a$ and $b$ are in different bipartite classes, they must be connected, and they cannot have another neighbor.
\end{proof}

Now, from Theorem~\ref{theorem:rationaleigenvaluestree} we know that if perfect state transfer happens in a tree between $a$ and $b$ and they are in different bipartite classes, all the eigenvalues in $S^-$ are of the form $2/k$, for $k$ an odd positive integer. The corollary above shows that if there are no eigenvalues of this form with $k > 1$, then unless the graph is $P_2$, there is no perfect state transfer between vertices in different bipartite classes. Calculation done in SAGE, using CoCalc, showed no trees with eigenvalues of this form for trees with up to $16$ vertices.

\section{Conclusion}

Two vertices are cospectral if and only if, at all times, the amplitude of survival probability in a quantum walk is the same for both. We have shown that when having the normalized Laplacian acting as Hamiltonian, cospectrality is equivalent to, in a classical uniform random walk, having the probability of return at both vertices after $k$ steps being the same for all $k$. 

This suggests that the normalized Laplacian is perhaps the Hamiltonian choice that will provide more similarities between classical random walks and quantum walks.

A basic quantum walk feature is the possibility of perfect state transfer. Our main result was a characterization of perfect state transfer in this model. We borrow from known methods, but we overcome an apparent extra difficulty when dealing with non-integral matrices.

On the matter of state transfer in the normalized Laplacian model in paths, it is shown in \cite{alvir2016perfect} that perfect state transfer occurs if and only if the path has $2$ or $3$ vertices. This is done upon showing an equivalence, by means of an equitable partition, to state transfer between antipodal vertices of cycles in the adjacency matrix model. We note that in conjunction with a result in \cite{PalCirculantPGST}, one obtains that pretty good state transfer in paths according to the normalized Laplacian occurs if and only if the path contains $2^k+1$ vertices, for $k \geq 0$. We have independently obtained this result using elementary methods. The connection with the known fact \cite{CoutinhoBanchiGodsilSeve} that pretty good state transfer in paths with the combinatorial Laplacian occurs if and only if the path has $2^k$ vertices cannot be ignored.

For trees, we were not able to show a complete analogous result to the combinatorial Laplacian case --- that is, that perfect state transfer does not occur. We were however able to significantly restrict the possible cases, and therefore we conjecture that apart from the paths on $2$ and $3$ vertices, no tree admits perfect state transfer according to the normalized Laplacian.

\section*{Acknowledgements}

Pedro Baptista acknowledges the scholarship UFMG-ADRC and the Master's scholarship from CAPES.

\bibliographystyle{unsrt}
\bibliography{qwalk.bib}

\end{document}